\documentclass{article}
\usepackage[utf8]{inputenc}

\usepackage{amsthm}
\usepackage{amsmath}
\usepackage{amssymb}

  \newtheorem{theorem}{Theorem}

  \newtheorem{Proposition} {Proposition}

\theoremstyle{definition}

 \newtheorem{definition} {Definition}

  \newtheorem{remark}{Remark}

\begin{document}
\title{An RIP-based approach to $\Sigma\Delta$ quantization for compressed sensing}

\author{Joe-Mei Feng and Felix Krahmer}
\maketitle
\begin{abstract}
 In this paper, we provide a new approach to estimating the error of reconstruction from $\Sigma\Delta$ quantized compressed sensing measurements.
 Our method is based on the restricted isometry property (RIP) of a certain projection of the measurement matrix. 
 Our result yields simple proofs and a slight generalization of the best-known reconstruction error bounds for Gaussian and subgaussian measurement matrices. 
\end{abstract}

\section{Introduction}

\subsection{Compressed sensing}\label{Compressed sensing}

 Compressed sensing has drawn significant attention since the seminal works by Cand\`es, Romberg, Tao  \cite{candes}, and Donoho \cite{do06-2}.
 The theory of compressed sensing is based on the observation that various cases of natural signals are approximately sparse with 
 respect to certain bases or frames.
 The basic idea is to recover such signals from a small number of linear measurements.
 Hence the problem turns into an underdetermined linear system.
Various criteria have been proposed to determine whether such a system has a unique sparse solution. 
In this paper we will work with the restricted isometry property (RIP) as introduced by Cand\`{e}s et al. \cite{ECTT} in the context of recovery guarantees for $\ell_1$ minimization.
\begin{definition}
  A matrix $A\in\mathbb{R}^{m\times N}$ has the restricted isometry property (RIP) of order $s$ if there exists $0<\delta<1$ such that for all $s$-sparse vectors 
  $x\in \mathbb{R}^N$, i.e., vectors that have at most $s$ non-zero components, one has
  \[
  (1-\delta)\|x\|_2^2\leq \|Ax\|_2^2\leq (1+\delta)\|x\|_2^2.
  \]
  The smallest such $\delta$ is called the restricted isometry constant of order $s$ and is denoted by  $\delta_s$.
\end{definition} 

There have been a number of works on recovery guarantees for compressed sensing with RIP measurement matrices. 
Recovery can be guaranteed for various algorithms. For the original context of $\ell_1$ minimization, 
the most recent results require the measurement matrix to have a restricted isometry constant of $\delta_{2s}<\frac{1}{\sqrt{2}}$ \cite{cai2013sparse}, 
which is known to be optimal \cite{davies2009restricted}.

Finding the restricted isometry constant of a measurement matrix is, in general, an NP hard problem \cite{tillmann2012computational}. 
On the other hand, deterministic matrix constructions with guaranteed RIP are only known for relatively large embedding dimensions (see for example \cite{de07-3}).
That is why many papers on the subject work with random matrices. 

Examples of random matrices known to have the RIP for large enough embedding dimension with high probability include subgaussian, 
partial random circulant \cite{Felix}, and partial random Fourier matrices \cite{rudelson2008sparse}. A subgaussian matrix has independent random entries whose tails are dominated by a Gaussian random variable (cf.~Definition~\ref{def:subgaussian}). 
Such matrices have been shown to have the RIP provided $m=\Omega(s\log(eN/s))$, see for example \cite{BDDW07:Johnson-Lindenstrauss}. 
This order of the embedding dimension $m$ is known to be optimal \cite{foucart2010gelfand}. Examples of subgaussian matrices include Gaussian and Bernoulli matrices.

  \subsection{Quantization}
 
To allow for digital transmission and storage of compressed sensing measurements,
one needs to quantize these measurements.  
  That is, the measurements need to be represented by finitely many symbols from a finite alphabet. In this paper, we only consider alphabets consisting of equispaced real numbers. 
  The extreme case of considering the set of only the two elements $\{-1,1\}$ is also called $1$-bit quantization. 
  
  The most intuitive method to 
 quantize the measurements is to map each of them to the closest element from the alphabet. Since this method
 processes the quantization independently for each measurement, it is also called memoryless scalar quantization (MSQ).
  
 Most of the literature on MSQ compressed sensing up to date considers 1-bit quantization 
\cite{Boufounos,JacquesLaska,Plan,Ai},
which amounts to considering only the measurement signs.   
Jacques et al. \cite{JacquesLaska} showed that for Gaussian measurements or 
measurements drawn uniformly from the unit sphere, a reconstruction error of $O(\frac{s}{m}\log\frac{mN}{s})$ is feasible. However, they did not provide an efficient algorithm that guarantees this accuracy.
Later, for Gaussian measurements, Gupta et al. \cite{Gupta}
demonstrated that one may tractably recover the support of a signal from $O(s\log N)$ measurements. 
Plan et al.~\cite{Plan} showed that one can, again for Gaussian measurements,
 reconstruct the direction of an $s$-sparse signal via convex optimization,
with accuracy $O((\frac{s}{m})^{\frac{1}{5}})$ up to logarithmic factors with high probability. 
Ai et al. \cite{Ai} derived similar results for subgaussian measurements under additional assumptions on the size of the signal entries. 

On the other hand, in \cite{JacquesLaska} it was shown that the $\ell_2$ reconstruction error can never be better 
than $\Omega(\frac{s}{m})$.
To break this bottleneck of MSQ,
$\Sigma\Delta$ quantization for compressed sensing has drawn attention recently. 
$\Sigma\Delta$ quantizes a vector as a whole rather than the components individually, 
i.e., the quantized values depend on previous quantization steps. 

$\Sigma\Delta$ quantization was originally introduced as an efficient quantizer for redundant representation of oversampled band-limited functions \cite{inose1962telemetering}.
Later on, a rigorous mathematical error analysis was provided by \cite{daubechies2003approximating} and many follow-up papers. The best known error decay rates are exponential in the oversampling rate, as derived in \cite{gunturk2003one, Deift}. This is known to be optimal: in \cite{KW12}, corresponding lower bounds are derived, which also show that the achievable accuracy must depend on the signal amplitude.

In \cite{benedetto2006sigma}, $\Sigma\Delta$ has been extended to frame expansions; this will be also the viewpoint taken in this paper. 
The first works on $\Sigma\Delta$ schemes for frame quantization, such as \cite{benedetto2006sigma, BPA2007}, required frame constructions with particular smoothness properties to yield reconstruction guarantees. In \cite{lammers:adf}, the  authors observed that what is needed is in fact a requirement on the dual frame used for reconstruction rather than the frame itself. 
Reconstruction guarantees can hence be improved by choosing the dual frame used for reconstruction appropriately. Optimizing the dual frame in this respect led to the definition of Sobolev dual frames \cite{blum2010sobolev}, cf.~Section~\ref{sobolevdual} below. Combined with the exponential error bounds derived for the corresponding $\Sigma\Delta$ schemes for bandlimited functions \cite{gunturk2003one, Deift}, Sobolev dual reconstructions yield root-exponential error decay in the oversampling rate. This constitutes the best known accuracy guarantees for coarse frame quantization, both for harmonic frames and Sobolev self-dual frames \cite{
KSW11} and for subgaussian random frames \cite{KrahmerSaab}. 

Sobolev dual reconstructions have also been crucial for being able to apply $\Sigma\Delta$ quantization to compressed sensing measurements. G{\"u}nt{\"u}rk et al.~\cite{Rayan} proved the first recovery guarantees for this setup, showing that for $r$th order $\Sigma\Delta$ quantization
applied to Gaussian compressed sensing measurements, the $\ell_2$ reconstruction error is of order $O((\frac{s}{m})^{\alpha(r-\frac{1}{2})})$ with high probability. Here $\alpha\in(0,1)$ is a parameter and the required measurements grows with 
$\alpha$, tending to infinity as $\alpha\rightarrow 1$. Indeed for $r$ large enough this breaks the MSQ bottleneck. 
More recently, in \cite{KrahmerSaab}, this result has been generalized to subgaussian measurements.

 \subsection{Contributions}

The main contribution of this paper is that the restricted isometry property (RIP) is applied to 
estimate the error bound for $\Sigma\Delta$ quantized compressed sensing. 
That is, once we know the restricted isometry constant of a modification of the measurement matrix, we 
can estimate the reconstruction error.

In the following results, we assume that the $\Sigma\Delta$ quantized measurements with quantization alphabet 
$\mathcal{Z}=\Delta \mathbb{Z}$, $\Delta>0$, are given. We refer the readers to Section \ref{quantization} for details
on the quantization scheme employed.
A special role is played by the $r$th power of the inverse of the finite difference matrix $D$ as introduced in \eqref{D} below;
denoting the singular value decomposition of $D^{-r}$ by $D^{-r}=U_{D^{-r}}S_{D^{-r}}V^*_{D^{-r}}$,
we obtain our main theorem given as follows.

 \begin{theorem}\label{error}
 Suppose one is given a measurement matrix $\Phi\in \mathbb{R}^{m\times N}$ such that both $\Phi$ and $\sqrt{\frac{1}{\ell}}P_\ell V^*_{D^{-r}}\Phi$, $\ell\leq m$ have the restricted isometry constant $\delta_{2s}<\frac{1}{\sqrt{2}}$,
where $P_\ell$ maps a vector to its first $\ell$ components.
 
 Then for an $s$-sparse signal $x\in\mathbb{R}^N$ satisfying $\min_{j}|x_j|\geq K 2^{r-\frac{1}{2}}\Delta$, for some positive constant $K$,  denote by $q$ the $r$th order $\Sigma\Delta$ quantized measurements of $\Phi x$ with step size $\Delta$.
Furthermore, denote by $T$ the support set recovered from $\Phi x$ via $\ell_1$ minimization and choose $L_{sob,r}$ to be the Sobolev dual matrix of $\Phi_T$ (see Section \ref{sobolevdual} for details).
Then reconstructing the 
signal via $\hat{x}_T=L_{sob,r}q$
yields a reconstruction error bounded by 
\[
\|x-\hat{x}\|_2 \leq C\Delta (\frac{m}{\ell})^{-r+\frac{1}{2}},
\]
where $C>0$ is a constant depending only on r.
\end{theorem}
Note from Theorem \ref{error} that smaller values of $\ell$ yield better error bounds.
However, $\ell$ has to be large enough such that $\frac{1}{\sqrt{\ell}}(P_\ell V^*_{D^{-r}}\Phi)$ has the restricted isometry constant $\delta_{2s}\leq\tfrac{1}{\sqrt{2}}$.

This result can be applied to obtain recovery guarantees for Gaussian and subgaussian measurements (in the sense of Definition \ref{def:subgaussian} below).
The resulting bounds for the first two cases agrees with those derived in \cite{Rayan} and \cite{KrahmerSaab}, as summarized in the following Theorem. 
\begin{theorem}[\cite{Rayan, KrahmerSaab}]\label{subgaussian}

 Let $\Phi$ be an $m\times N$ matrix whose entries are independent, mean zero, unit variance $\rho$-subgaussian random variables and suppose that 
 $\lambda:=m/k\geq(C\log(eN/k))^{\frac{1}{1-\alpha}}$ where $\alpha\in(0,1).$ With high probability the $r$th order $\Sigma\Delta$ reconstruction $\hat{x}$ satisfies
 \begin{equation*}
  \|x-\hat{x}\|_2\leq C'\lambda^{-\alpha(r-1/2)}\delta,
 \end{equation*}
 for all $x\in\Sigma_k^N$ for which $\min_{j\in \mbox{\rm{supp}}(x)}|x_j|> K'\Delta.$ Again, $\Delta$ is the step size of the $\Sigma\Delta$ quantization alphabet and $C, C', K'$ are appropriate constants that depend
 only on $r$ and $\rho$.

\end{theorem}

 \subsection{Organization}
The paper is organized as follows.
We first introduce in Section~\ref{Background} some background and previous results on $\Sigma\Delta$ quantization, suprema of chaos processes, 
and the partial random circulant matrices.
In Section \ref{RIP-based error analysis} we present our main result showing how
the RIP is used to estimate the reconstruction error for quantized compressed sensing.
In Section \ref{Gaussian and subgaussian matrices}, we explain how our result recovers the best-known bounds
for Gausssian and subgaussian measurement matrices using a simple argument, in Section \ref{generalization} we slightly generalize these bounds.
We conclude in Section~\ref{conclusion}.

\section{Background and previous results}\label{Background}
\subsection{Notation}\label{Notations}
Throughout this paper, we use the following notation.
The set $D_{s,N}=\{x\in\mathbb{R}|\|x\|_2\leq 1,\|x\|_0\leq s\}$ is the set of unit norm $s$-sparse vectors.
The $\ell_0$-norm
$\|\cdot\|_0$ counts the number of non-zero components of a vector. Given a signal $x$, the support set of $x$, in short, supp $x$, 
is the index set of the non-zero components.
The $\ell_2$-operator norm is denoted by $\|A\|_{2\rightarrow 2}=\sup_{\|x\|_2= 1}\|Ax\|_2$.
For a matrix $A$, $\sigma_i(A)$ and $\sigma_{\min}(A)$ denote the $i$th largest and the smallest singular value, respectively.
Furthermore we write $\gtrsim$ and $\lesssim$ to denote $\geq$ or $\leq$ up to a positive multiplicative constant.
The Moore-Penrose pseudoinverse of a matrix $A$ is denoted by ${A}^{\dag}=(A^*A)^{-1}A^*$.

We will mainly study subgaussian random matrices, that is, matrices with independent subgaussian entries in the sense of the following definition.
\begin{definition}\label{def:subgaussian}
A random variable $X$ is called $\rho$-subgaussian if $\mathbb{P}(|X|\geq t)\leq 2\exp(-t^2/2{\rho}^2)$.
\end{definition}
\subsection{$\Sigma\Delta$ Quantization}\label{quantization}


In this paper, we exclusively focus on quantization alphabets $\mathcal{Z}$ such that $\mathcal{Z}=\Delta\mathbb{Z}$, for some $\Delta>0$. 
Note that while this is an infinite set, one can show that in fact only a finite range of values are assumed \cite{Rayan, KSW11}.
Hence this setup is in line with requiring a finite alphabet.
The idea of $r$th order $\Sigma\Delta$ quantization is to quantize each component of a vector taking the previous $r$ quantization steps into account.
More explicitly,
a greedy $r$th order $\Sigma\Delta$ quantization scheme maps a sequence of inputs $(y_j)$ to elements $q_i\in\mathcal{Z}$
via an internal state variable $u_i$ chosen to satisfy the recurrence relation
\begin{align}\label{sigmadelta}
(\Delta^r u)_i:=\sum_{j=0}^{r} \binom{r}{j}(-1)^{j}u_{i-j}=y_i-q_i,
\end{align}
where $q_i$ is chosen such that $|u_i|$ is minimized (Note that only in this equation, $\Delta$ denotes the finite difference operator, whereas all other occurences in this paper refer to the quantization step size).  

With the initial condition $(u_i)_{i=0}^{-\infty}=0$, Equation (\ref{sigmadelta})
can be expressed as
\begin{equation*}
D^{r} u=y-q,
\end{equation*}
where the finite difference matrix $D\in{\mathbb R}^{m\times m}$ is given by 
\begin{equation}\label{D}
D_{ij}\equiv\left\{\begin{array}{ll}
1&,\mbox{ if }i=j,\\
-1&,\mbox{ if }i=j+1,\\
0&,\mbox{ otherwise.}
\end{array}\right.
\end{equation}

\subsubsection{Support set recovery}
Given an $s$-sparse signal $x$, and an $m\times N$ measurement matrix $\Phi$, where $m\ll N$, we acquire measurements $y=\Phi x$.
Applying an $r$th order $\Sigma\Delta$ quantization scheme to $y$, we obtain $q$.
Treating $q$ as perturbed measurements, i.e., $q=y+e=\Phi x+e$, one can determine the support set. 
This is a consequence of the following observation, which is
a modified version of Proposition 4.1 in \cite{Rayan} combined with the reconstruction guarantees in \cite{cai2013sparse}.
\begin{Proposition}\label{coarserecovery}
Given $\epsilon>0$ as well as $x\in\mathbb{R}^N$ an $s$-sparse signal with $\operatorname{supp} x = T$ and 
$\min_{j\in T}|x_j|\geq K \frac{\epsilon}{\sqrt{m}}$. Here $K$ is an absolute constant.  Let $\Phi\in\mathbb{R}^{N\times m}$ be a measurement matrix such that $\frac{1}{\sqrt{m}}\Phi$ has the RIP with $\delta_{2s}<\frac{1}{\sqrt{2}}$. Denote by $e\in {\mathbb R}^m$ a noise vector with $\|e\|_2\leq\epsilon$, and let $x'$ be the signal reconstructed from the noisy measurements $q=\Phi x+e$ via $\ell_1$ minimization, i.e., 
\begin{equation*}
x'=\mbox{\rm{arg}}\min \|z\|_1\ \mbox{subject to } \|\Phi z-q\|_2\leq\epsilon.
\end{equation*}
Then 
the index set of largest $s$ components of $x'$ is $T$, that is, the support set of $x$ is correctly recovered.
\end{Proposition}

Note that in this result, the measurement matrix $\Phi$ is not normalized, while
in the compressed sensing literature, it is common to normalize the measurement matrix such that it has unit-norm columns.
This is because for normalized matrix columns, each measurement will be of order $\frac{1}{\sqrt{m}}$, 
so quantizing it with a fixed step size $\Delta$ will lead to worse and worse resolution.
To allow for a fair comparison when $m$ grows, the measurements should rather be chosen independently of $m$. 
Therefore, in this paper as well as in \cite{Rayan} the measurement matrices
are not normalized, each entry of the measurement matrices is chosen to have variance one. 

To apply Proposition \ref{coarserecovery} to greedy $\Sigma\Delta$ quantization, 
one sets $e=q-y$, where $q$ is the quantized measurement vector. Elementary estimates (cf. \cite{Rayan}) yield that $\|q-y\|_2\leq 2^{r-1}\Delta\sqrt{m}$.
Thus one obtains that $\ell_1$ minimization recovers the correct support set provided that $\frac{1}{\sqrt{m}}\Phi$ has restricted isometry constant $\delta_{2s}<\frac{1}{\sqrt{2}}$
and $\min_{j}|x_j|\geq K 2^{r-\frac{1}{2}}\Delta$.

\subsubsection{Estimating the error and the  Sobolev dual}\label{sobolevdual}
When the support set $T$ has been identified, we solve for $x$
using some left inverse of $\Phi_T$, say $L$. 
Then the reconstruction $\ell_2$-error is given by
\begin{align*}
\|x-\hat{x}\|_2&=\|Ly-Lq\|_2=\|L(y-q)\|_2\\
 &=\|L(D^{r}u)\|_2\leq\|LD^{r}\|_{2 \rightarrow 2}\|u\|_2.
\end{align*}
The Sobolev dual matrix $L_{sob,r}$, first introduced in \cite{blum2010sobolev}, is a left inverse of $\Phi_T$ defined to minimize $\|LD^r\|_{2 \rightarrow 2}$, i.e.,
\begin{equation*}
\begin{array}{cc}
L_{sob,r}=\mbox{\rm{arg}}\min_{L}\|LD^r\|_{2 \rightarrow 2}&\ \mbox{subject to}\ L{\Phi_T}=I. 
\end{array}
\end{equation*}
The geometric intuition is that this dual frame is smoothly varying. 

As in \cite{Rayan}, the explicit formula $L_{sob,r}D^r=(D^{-r}{\Phi_T})^\dag$ yields the error bound
 \begin{align}
\|x-\hat{x}\|_2&\leq \|(D^{-r}{\Phi_T})^\dag\|_{2 \rightarrow 2}\|u\|_2 \nonumber\\
&=\frac{1}{\sigma_{min}(D^{-r}{\Phi_T})}\|u\|_2\leq
\frac{\Delta\sqrt{m}}{2\sigma_{min}(D^{-r}{\Phi_T})},\label{main}
\end{align}
where the last inequality is derived in \cite{Rayan}.

A key ingredient to bounding $\sigma_{\min}(D^{-r}\Phi_T)$ is the following result from the study of Toeplitz matrices, 
which depends heavily on Weyl's inequality \cite{horn2012matrix} (see for example \cite{Rayan}).
\begin{Proposition}\label{D^-r}
Let $r$ be any positive integer and $D$ be as in (\ref{D}). There are positive constants $c_{s_1}(r)$ and $c_{s_2}(r)$, independent of $m$, such that
\begin{equation*}
c_{s_1}(r)(\frac{m}{j})^r\leq \sigma_j(D^{-r})\leq c_{s_2}(r)(\frac{m}{j})^{r},\ j=1,\ldots,m.
\end{equation*}
\end{Proposition}

\section{RIP-based error analysis}\label{RIP-based error analysis}


In this section we will give the quantized compressed sensing problem a mathematical model, 
and explain how we approach the reconstruction error via the RIP.
In the next two sections we show its applications.
From Section \ref{sobolevdual}, the main issue to estimate the reconstruction error is to estimate $\sigma_{min}(D^{-r}{\Phi_T})$.
Finding the supremum of this expression over all potential support sets $T$ can be interpreted as finding 
the supremum of the smallest image under $D^{-r}\Phi$ over all unit norm $s$-sparse vectors.
This motivates the connection to the RIP.

In the following proof we show how the RIP can be applied to find this effective smallest singular value.

 \begin{proof}[Proof of Theorem \ref{error}]
 As the assumptions of the theorem are stronger than those of Proposition \ref{coarserecovery}, we conclude that the support is correctly recovered. Based on this observagtion, we now show the error bound.
 Recall that $D^{-r}=U_{D^{-r}}S_{D^{-r}}V_{D^{-r}}^*$.
 Then, as $S$ is a diagonal matrix,

 \begin{align}
 \sigma_{\min}(D^{-r}{\Phi_T})&=\sigma_{\min}(S_{D^{-r}}V^*_{D^{-r}}{\Phi_T})\notag\\
 &\geq \sigma_{\min}(P_\ell S_{D^{-r}}V^*_{D^{-r}}{\Phi_T})\notag\\
 &= \sigma_{\min}((P_{\ell}S_{D^{-r}}P^*_{\ell} )(P_\ell V^*_{D^{-r}}{\Phi_T}))\notag\\
 &\geq s_\ell \sigma_{\min}(P_\ell V^*_{D^{-r}}{\Phi_T})\notag\\
 &\gtrsim  (\frac{m}{\ell})^r\sigma_{\min}(P_\ell V^*_{D^{-r}}{\Phi_T})\label{minsingular},
 \end{align}
 where the final inequality follows from Proposition \ref{D^-r}.
 
Thus we need to bound $\sigma_{\min}(P_\ell V^*_{D^{-r}}{\Phi_T})$ uniformly over all possible support sets $T$.
 Indeed by the RIP assumption for $\frac{1}{\sqrt{\ell}}P_\ell V^*_{D^{-r}}{\Phi}$, we obtain 
 that $\sigma_{\min}(P_\ell V^*_{D^{-r}}{\Phi_T})$ is uniformly bounded from below by 
 \begin{equation}\label{subRIP}
 \sqrt{\ell}\sqrt{1-\tfrac{1}{\sqrt{2}}}.
 \end{equation}
 The theorem  follows by combining \eqref{main}, (\ref{minsingular}), and (\ref{subRIP}).

\end{proof}

 \section{Gaussian and subgaussian matrices}\label{Gaussian and subgaussian matrices}
 To illustrate the simplicity of our method, we first present a proof of Theorem \ref{subgaussian} for standard Gaussian matrices, i.e. matrices with independent entries $\Phi_{i,j}\sim \mathcal{N}(0,1)$.
 \begin{proof}[Proof of Theorem \ref{subgaussian} for Gaussian matrices]
 Set $\ell:=m (\frac{s}{m})^\alpha$. As the second factor is less than $1$, one always has $1\leq \ell\leq m$.
 Since $\Phi$ is a standard Gaussian random matrix, due to rotation invariance $\widetilde\Phi:=(P_\ell V^*_{D^{-r}}\Phi)$ is also a standard Gaussian random matrix. 
 The assumption on $\lambda$ implies that its embedding dimension satisfies $\ell \geq C s \log(\tfrac{eN}{k})$, so standard results (see, e.g., \cite{BDDW07:Johnson-Lindenstrauss}) yield that for $C$ and $C'$ large enough, both $\tfrac{1}{\sqrt \ell}\widetilde \Phi$ and $\tfrac{1}{\sqrt m}\Phi$ have the RIP with constant $\delta_{2s}\leq \tfrac{1}{\sqrt{2}}$.
Applying Theorem \ref{error}  (choose $K'=K2^{r-\frac{1}{2}}$), we obtain
  \begin{align*}
\|x-\hat{x}\|_2&\lesssim\Delta(\frac{m}{s})^{-\alpha (r-\frac{1}{2})},
\end{align*}
again with high probability, as desired.
\end{proof}

{\em Sketch of proof of Theorem \ref{subgaussian} for subgaussian matrices:}\\
The proof proceeds long the same lines as for Gaussian matrices, except that one cannot use the rotation invariance.
To bound the RIP constant of $\tfrac{1}{\sqrt \ell} P_\ell V_{D^{-r}}^* \Phi$, 
we note that for any $x\in D_{s,N}$, $\|P_\ell V_{D^{-r}}^*\Phi x\|_2^2$ is a quadratic form in the ``vectorization''
of $\Phi$. Hence its tail decay can be estimated via the Hanson-Wright inequality \cite{HW71,rv13}. The RIP then follows via a union bound over an $\epsilon$-net of $D_{s,N}.$ This approach is related to certain steps in the original proof in \cite{KrahmerSaab}.

\begin{remark}
Note that a complete proof for $\Sigma\Delta$ recovery guarantees needs both support set recovery and fine recovery (cf.~\cite{Rayan, KrahmerSaab}) and in this paper we omitted the details of the former.
 We argue, however, that this coarse recovery step is straightforwardly based on standard compressed sensing results.
 Hence the core of our error estimate is really just captured in a few lines.
\end{remark}

\section{Generalization}\label{generalization}
In contrast to the techniques presented in \cite{KrahmerSaab}, our method generalizes to certain random matrices with independent subgaussian columns, but no entrywise independence.
As an additional criterion, one needs a type of small ball condition for $P_\ell V_{D^{-r}}^*$ applied to one of the random columns of $\Phi$, (which denoted by $\Phi_j$ in the following).
That is, one needs to exclude that $\|P_{\ell}V_{D^{-r}}\Phi_j\|_2$ is small with too large probability.
If such a condition holds, the necessary RIP bound follows from a modified version of the RIP bound for matrices with independent subgaussian columns \cite{vershynin2010introduction}.
While we do not consider this to be an important generalization (which is why we refrain from presenting the details), we still believe it shows that our method is stronger than previous approaches,
so we see the potential to apply it to more relevant, structured measurement scenarios such as partial random Fourier matrices, partial random circulant matrices, etc.

\section{Conclusion}\label{conclusion}
In this work we provided a new technique for bounding the reconstruction error arising in $\Sigma\Delta$ quantization for compressed sensing.
In addition to greatly simplifying the proofs for the best known recovery guarantees, the new viewpoint hopefully opens the possibility to study broader classes of measurement matries.

\section*{Acknowledgment} The authors would like to thank Rayan Saab for providing helpful comments.
This work was supported by the DFG Research Training Group 1023.

\bibliographystyle{plain} 
\bibliography{references} 
\end{document}